\documentclass[11pt]{article}

\usepackage{graphicx}
\usepackage{dcolumn}
\usepackage{bm}

\usepackage{mathrsfs}  

\usepackage{appendix}
\usepackage{forest}
\usepackage{tikz}
\usepackage{qcircuit}

\usepackage[utf8]{inputenc}
\usepackage[T1]{fontenc}
\usepackage{color}
\usepackage{physics}
\usepackage{fullpage}
\usepackage{amsthm}
\usepackage{amssymb}
\usepackage{thm-restate}
\usepackage{mathtools}
\usepackage{authblk}
\usepackage{bbm}
\usepackage[colorlinks,linkcolor=blue,citecolor=blue,urlcolor=magenta,filecolor=cyan]{hyperref}
\usepackage[capitalize]{cleveref}

\newtheorem{theorem}{Theorem}
\newtheorem{lemma}[theorem]{Lemma}
\newtheorem{definition}[theorem]{Definition}

\newtheorem{remark}[theorem]{Remark}

\def\bbc{\mathbb{C}}

\newcommand{\comment}[1]{}

\title{Time-Dependent Hamiltonian Simulation via Time-Independent Dynamics in a Larger Space}

\author{Zecheng Li}
\author{Chunhao Wang}

\affil{Department of Computer Science and Engineering, Pennsylvania State University}
\affil[]{Email: \{zxl5523,cwang\}@psu.edu}

\date{}

\begin{document}
\maketitle

\begin{abstract}
In this paper, we present a proof-of-concept quantum algorithm for simulating time-dependent Hamiltonian evolution by reducing the problem to simulating a time-independent Hamiltonian in a larger space using a discrete clock Hamiltonian construction. A similar construction was first explored for this simulation problem by Watkins, Wiebe, Roggero, and Lee [PRX Quantum, 2024]. Our algorithm improves upon their work in terms of the dependence on evolution time and precision. In addition, the complexity matches the state-of-the-art simulation algorithms using other approaches. To achieve this improvement, we use Duhamel's principle to treat the clock and system Hamiltonians separately and exploit properties of Gaussian quadrature to reduce the simulation cost. Our approach demonstrates that time-dependent Hamiltonian simulation can be as efficient in a simpler framework and hence provides a new angle to model and simulate time-dependent systems.
\end{abstract}

\section{Introduction}
In the past few years, quantum computing has made remarkable progress, exemplified by milestones such as recent developments of quantum computers~\cite{AAB+19,AAA+24} and demonstrations of error-correcting codes on real quantum devices~\cite{SER+23,PdM+24,BDE+24}. These advancements are filling the gaps between theoretical and practical quantum advantages, particularly in solving problems that are intractable for classical computers. Among these, Hamiltonian simulation stands out as a critical application, and it is probably the first practical application for large-scale fault-tolerant quantum computers. The task of Hamiltonian simulation involves modeling the time evolution of a quantum system, governed by a Hamiltonian $H$, through implementing the unitary operator $e^{-iHt}$ for evolution time $t$ using elementary quantum gates. The ability to simulate the time evolution of Hamiltonians enables the study of complex phenomena in chemistry, material science, and fundamental physics. For instance, simulating molecular interactions and the electronic structure of materials---tasks that are computationally prohibitive for classical computers---can be efficiently tackled using quantum computers. 

Since 1982, when Feynman first proposed quantum computers~\cite{Feynman82}, efficient quantum algorithms for Hamiltonian simulation have been extensively studied. Some examples include~\cite{Lloyd96,ATS03,BACS07,BC09,Childs10,WBHS11,CW12,BCCKS15,BCK15,BN16,BCCKS17}. In the recent decade, advanced techniques such as quantum signal processing~\cite{LC17a}, qubitization~\cite{LC16}, and quantum singular value transformation~\cite{GSLW19} have pushed the complexity of Hamiltonian simulation algorithms to optimal. 
In the meantime, a considerable amount of efforts have been made to study and design simulation techniques from a practical perspective~\cite{CMNRS18,Campbell19,BCDGMS24}. Furthermore, Hamiltonian serves as the backbone for many advanced quantum algorithms such as quantum phase estimation~\cite{Kitaev95}, quantum linear systems solvers~\cite{HHL09}, and variational quantum eigensolvers~\cite{PMS+14}, which are pivotal for applications ranging from drug discovery to optimization problems. As quantum computing continues to advance, the ability to simulate Hamiltonians will unlock transformative possibilities across different fields, solidifying its role as a cornerstone of the quantum revolution.

Time-dependent Hamiltonian simulation extends the framework of Hamiltonian simulation to systems where the Hamiltonian $H(t)$ varies with time, reflecting more realistic scenarios in quantum dynamics. Unlike time-independent Hamiltonians, time-dependent Hamiltonians require more sophisticated techniques to accurately simulate the time-ordered exponential 
\begin{align}
    U(t) = \mathcal{T} \exp(-i \int_0^t H(\tau) \dd \tau), 
\end{align}
 where $\mathcal{T}$ denotes the time-ordering operator. This capability is crucial for modeling a wide range of physical phenomena, such as driven quantum systems, chemical reactions under external fields, and quantum control processes. Recent developments in quantum algorithms for simulating time-dependent Hamiltonians, such as simulation algorithms based on Dyson series~\cite{KSB19} and the interaction picture~\cite{LW19}, have significantly improved the efficiency and precision of these simulations. Berry, Childs, Su, Wang, and Wiebe~\cite{BCSWW20} further improved the simulation techniques by considering the $L^1$-norm dependence, which is much more efficient when the Hamiltonian varies significantly. Recently, Mizuta and Fujii~\cite{MF23} showed that when the Hamiltonian is time-periodic, the simulation algorithm can be made optimal---saturating the simulation lower bound for simulating time-independent Hamiltonians. These developments are particularly important for applications in quantum chemistry, where time-dependent electric or magnetic fields influence molecular dynamics.

It is intriguing to look at time-dependent Hamiltonians from a different angle: we can view these systems as open quantum systems as the system itself is controlled by some external field that enables the time-dependence, while the external control is sophisticated enough to maintain the unitarity of the system, preventing any energy or information leakage to the environment. If the whole universe is running as a unitary process, then every open quantum system is a subsystem of a unitary process in a larger space, including time-dependent Hamiltonian evolution. From this perspective, it is natural to study time-dependent Hamiltonian simulation algorithms as a time-independent Hamiltonian evolution in a larger space. 

To construct such a time-independent Hamiltonian, we consider the discrete clock construction, where the external register is in a ``clock'' state controlling the system Hamiltonian to act for the corresponding time. 
There is also a clock Hamiltonian that advances the clock state. Intuitively, this construction incorporates the external control coherently in a larger system, and the reduced system dynamics is the time-dependent Hamiltonian evolution. Informally, for a time-dependent Hamiltonian $H(t)$ for $t \in [0, T]$, the joint time-independent system Hamiltonian, denoted by $H_{\mathrm{sys}}$, is defined as
\begin{align}
    H_{\mathrm{sys}} = \sum_{t} \ketbra{t}{t} \otimes H(t).
\end{align}
Here, the first register is the ``clock'' register, and the clock Hamiltonian, denoted by $H_{\mathrm{clk}}$, is defined such that
\begin{align}
  e^{-i H_{\mathrm{clk}}\tau} \ket{t}\ket{\psi} = \ket{t+\tau}\ket{\psi}.
\end{align}
Note that in the above description, the clock space is continuous, which cannot be directly used in computation. In \cref{sec:clock}, we present the details of the clock Hamiltonian construction and its discretization. 

The idea of the clock Hamiltonian construction is not new. It dates back to 1985 when Feynman~\cite{Feynman85} introduced the concept of encoding computation histories into Hamiltonians. In 2002, Kitaev used this construction to prove QMA-completeness~\cite{KSV02}. 
In the context of Hamiltonian simulation, Watkins, Wiebe, Roggero, and Lee used this construction in their pioneering work~\cite{WWRL24} to present an efficient quantum algorithm for simulating time-dependent Hamiltonians. The complexity of their algorithm scales quadratically in the evolution time and polynomially in the inverse of precision. Their suboptimal complexity is mainly due to the unfavorably large norm of the clock Hamiltonian. Although not matching the state-of-the-art time-dependent Hamiltonian simulation algorithm, their work was the first to show a proof of concept of simulating time-dependent Hamiltonians using the discrete clock. Recently, Mizuta, Ikeda, and Fujii~\cite{MIF24} considered product formula and multi-product formula for simulating time-dependent Hamiltonians, and they improved the error bound of these methods with commutator scaling. It is worth noting that although simulation methods based on product formula and multi-product formula do not match the asymptotic performance of post-Trotter techniques (including the simulation algorithm presented in this paper), these methods are usually more practical to implement in near-term quantum devices thanks to their much favorable requirements on the number of ancillary qubits.

In this paper, we propose a time-dependent Hamiltonian simulation algorithm that improves upon~\cite{WWRL24}. In particular, we use a rather simpler discrete-clock construction with different simulation techniques to achieve polynomial improvement in the evolution time and exponential improvement in the precision dependence. Let $H(t)$ for $0 \leq t \leq T$ be the time-dependent Hamiltonian to simulate. We consider the standard sparse-access model \cite{LW19} for time-dependent Hamiltonians of dimension $N$ which assumes that the number of nonzeros entries in each row is at most $d$ in the entire time interval $[0, T]$. We have access to the following oracles:
\begin{align}
\label{oracle_Ham}
    \mathcal{O}_s \ket{j, s} &= \ket{j, \mathrm{col}(j, s)}, \\
    \mathcal{O}_v \ket{t, j, k, z} &= \ket{t, j, k, z\oplus H_{jk}(t)},
\end{align}
where $\mathrm{col}(j, s)$ is the index of the $s$-th entry of row $j$ that can be nonzero in $[0, T]$.

Our main result is summarized in the following theorem.

\begin{theorem}
    \label{thm:main-theorem}
    Let $H(t)$ be a $d$-sparse time-dependent Hamiltonian on $n$ qubits for $t \in [0, T]$. Then the evolution of $H(T)$ for time $T$ can be simulated with precision $\epsilon$ using 
    \begin{align}
        O\left(dTH_{\max}\frac{\log(TdH_{\max}/\epsilon)}{\log\log(TdH_{\max}/\epsilon)}\right)
    \end{align}
    queries to $\mathcal{O}_s$ and $\mathcal{O}_v$, with 
    \begin{align}
        O\left(dTH_{\max}\left(n +
    \log\left(\frac{TdH_{\max}\dot H_{\max}}{\epsilon}\right)\cdot\frac{\log(TdH_{\max}/\epsilon)}{\log\log(TdH_{\max}/\epsilon)}\right)\right)
    \end{align}
    additional 1- and 2-qubit gates, and
    \begin{align}
        O\left( \frac{\log^2(TdH_{\max}/\epsilon)}{\log\log (Td H_{\max}/\epsilon)}  + \log \left(\frac{T}{\epsilon}\dot H_{\max}\right)\right)
    \end{align}
    ancillary qubits.
\end{theorem}
Here, $H_{\max}$ and $\dot H_{\max}$ are measures of the magnitude of $H(t)$ and the rate of change of $H(t)$, whose definitions are presented in \cref{sec:notation}.

\begin{remark}
    Although the complexity of our algorithm is presented in terms of the worst case of the max-norm of $H(t)$, our method can be adapted to achieve the $L^1$-norm dependence as in~\cite{BCSWW20}, assuming the ability to rescale the Hamiltonian. 
\end{remark}

\paragraph{Technical overview} The challenge of dealing with the discrete clock Hamiltonians is their large norms. This large norm is inevitable because as the discretization becomes finer, the clock state should evolve faster to catch up with the evolution pace. In addition, we use a rather simpler discrete-clock construction than~\cite{WWRL24}, and this simpler construction results in even larger norms of the clock Hamiltonian. To tackle this difficulty, we use an approximation scheme inspired by Duhamel's principle to separately treat the clock and system Hamiltonians. Duhamel's principle has been used in both quantum and classical algorithms for simulating open quantum systems~\cite{CL25,LW23}. We also use a scaled Gaussian quadrature to efficiently approximate the integral (obtained by Duhamel's principle) as a linear combination of unitaries (LCU). This involves a quadrature method that yields a much more efficient LCU approximation without dealing with the time ordering. However, we still need to deal with another difficulty that prevents us from achieving a poly-logarithmic dependence on precision: the norm of the higher-order derivative of the time-dependent evolution operator is much larger than the time-independent case (which was also observed in simulating time-dependent open quantum systems~\cite{HLLLWW24}). This norm is critical in the error bound of the Gaussian quadrature approximation. To tackle this, we explore efficient approximation schemes of the partial sum of quadrature weights so that the state preparation procedure of LCU can be efficiently implemented by the Grover-Rudolph approach, hence keeping the complexity $\mathrm{polylog}(1/\epsilon)$.

\paragraph{Related work using the interaction picture and truncated Dyson series} Our algorithm shares a similar structure with the simulation algorithm by Low and Wiebe~\cite{LW19}. In particular, our treatment by Duhamel's principle in \cref{thmt@@duhamel} can be derived by the interaction picture, and then the evolution operator may be implemented by a truncated Dyson series, as in~\cite{LW19}. From this perspective, our algorithm does not offer any complexity advantage compared with~\cite{LW19} and~\cite{KSB19}. However, the goal of this work is to demonstrate that a reduction from time-dependent Hamiltonian simulation to time-independent Hamiltonian simulation can also provide a near-optimal quantum algorithm for this problem, which, along this line, improves the previous clock construction in~\cite{WWRL24}. 
Here, we show that simulation algorithms on top of this reduction can be as efficient as the state-of-the-art simulation algorithms, hence providing a different yet efficient model of understanding time-dependent evolution by viewing it as time-independent evolution in a larger space. Note that this is not immediately obvious using the original interaction picture approach in~\cite{LW19}, where we need to simulate another time-dependent evolution as a subproblem.

\paragraph{Open questions} This research leaves several interesting open questions:
\begin{enumerate}
  \item There are generally two approaches to constructing clock states and clock operators. The first approach, as used in this paper, is based on taking the logarithm of an increment operator, which produces a simple clock operator but at the expense of having a large operator norm. It also suffers from sizable off-diagonal commutators with $H_{\mathrm{sys}}$. We managed to solve these problems by applying Duhamel's principle and Gaussian quadrature. The second approach, as in \cite{WWRL24}, inspired by the momentum operator, allows for truncation to control the operator norm; however, its associated Gaussian clock state must have a very tightly constrained width, which in turn restricts its stability and further decrease the simulation precision. An open question is whether there exists an alternative construction of clock operators and clock states that strikes a better balance between maintaining a stable and simple state and keeping the operator norm sufficiently small to be practical.
  \item What is the optimal cost for simulating time-dependent Hamiltonians? For time-independent Hamiltonians, the lower bound on the simulation cost is $\Omega(t + \log(1/\epsilon))$, and it has been saturated by advanced simulation techniques such as quantum signal processing~\cite{LC17a}, quantum singular value transformation~\cite{GSLW19}, and qubitization~\cite{LC16}. Since time-dependent Hamiltonians are more general, this additive lower bound also applies to the time-dependent case. Despite the known multiplicative upper bound, the additive lower bound was only achieved in special cases such as time-periodic Hamiltonians~\cite{MF23}. Both~\cite{WWRL24} and this work propose an avenue to fill the gap by modeling time-dependent systems as a time-independent system in a larger space; however, we still need technical tools to optimally simulate the time-dependent Hamiltonian in a larger space, especially to deal with the large norm.
\end{enumerate}

\section{Preliminaries}
\subsection{Notation}
\label{sec:notation}
For vectors, we use $\norm{\cdot}$ to denote the Euclidean norm. Let $H$ be a matrix, we use $H_{j,k}$ to denote its entry indexed by $(j, k)$. Define the \emph{max-norm} of a Hamiltonian $H$, denoted by $\norm{H}_{\max}$ as $\norm{H}_{\max} \coloneqq \max_{i, j} |H_{j, k}|$. For a time-dependent matrix $H(t)$, we define the following to measure its magnitude and rate of change.
$H_{\max} \coloneqq \max_{t\in [0,T]}\|H(t)\|_{\max}$, $\dot {H}(t) \coloneqq \frac{\mathrm{d}}{\mathrm{d}t} H(t)$, and $\dot H_{\max} \coloneqq \max_{t\in [0,T]}\left\| \frac{\mathrm{d}H(t)}{\mathrm{d}t}\right\|$.

\subsection{Block-encoding and the LCU construction}
Although we use sparse-access oracles as input, it is convenient to use block-encoding to present our algorithm. A block-encoding of a matrix $A$ is a unitary $U$ whose first block is a normalized $A$ in the sense that
    $U = \left(\begin{smallmatrix}
        A/\alpha & \cdot\\
        \cdot & \cdot
    \end{smallmatrix}\right)$.
More formally, if $A$ is acting on $n$ qubits, we say that an $(n+m)$-qubit unitary $U$ is an $(\alpha, m, \epsilon)$-block-encoding of $A$ is
\begin{align}
    \norm{A - \alpha (\bra{0} \otimes I_{2^n})U(\ket{0}\otimes I_{2^n})} \leq \epsilon,
\end{align}
where $I_{2^n}$ is the identity operator acting on $n$ qubits. When the number of qubits and precision are less important, we also refer to this block-encoding as an $\alpha$-block-encoding.

We adopt the following definition of the block-encoding for time-dependent matrices.
\begin{definition}[Time-dependent matrix block-encoding \cite{LW19}]
\label{oracle_inputblockencoding}
    Given a matrix $H(s): [0,t] \mapsto \mathbb{C}^{2^{n \times n}}$, integer $M > 0$ is the dimension of our discretized time space, and a promise $\|H\| 
    \leq \alpha$, then the time-dependent matrix block-encoding $\mathrm{HAM_T}\in \mathbb{C}^{M2^{n}\times M2^{n}}$ is defined as 
    \begin{align}
        \mathrm{HAM_T} &= \begin{pmatrix}
            H/\alpha & \cdot \\
            \cdot &\cdot
        \end{pmatrix},
        \quad H = \mathrm{Diagonal}[H(0), H(t/M), \dots, H((M-1)t/M)], \\
        &\implies (\langle0| \otimes I)\mathrm{HAM_T}(|0\rangle\otimes I) = \sum_{m=0}^{M-1}|m\rangle \langle m| \otimes \frac{H(mt/M)}{\alpha}.
    \end{align}
\end{definition}
For this work, it suffices to choose 
\begin{align}
    M = \Theta \left(\max_t \|\dot{H}\| \frac{T^3}{\epsilon^2} \right).
\end{align}

Here, $T$ is the total evolution time. The following lemma shows how to efficiently implement the block-encoding of time-dependent matrices using their sparse-access oracles.
\begin{lemma}[Preparing $\mathrm{HAM_T}$ with sparse access model \cite{LW19}]
\label{lemma_constructingbe}
    Given sparse access input oracle for time-dependent Hamiltonian $H(t)\in \mathbb{C}^{2^n \times 2^n}$ as in \cref{oracle_Ham},
    one can prepare $\mathrm{HAM_T}$ as defined in \cref{oracle_inputblockencoding}, with $O(1)$ queries of $\mathcal{O}_s$ and $\mathcal{O}_v$, and $O(n)$ primitive gates, with block-encoding normalizing constant $\alpha$ being $dH_{\max}$.
\end{lemma}

Our algorithm depends on implementing linear combinations of block-encodings. This construction, first formulated in~\cite{LW23a}, is a slight generalization of linear combinations of unitaries.
\begin{lemma}[Linear combinations of block-encodings~{\cite[Lemma 5]{LW23a}}]
  \label{lemma:sum-to-be}
  Suppose $A \coloneqq \sum_{j=1}^m y_j A_j \in \bbc^{2^n\times 2^n}$, where $A_j \in \bbc^{2^n \times 2^n}$ and $y_j > 0$ for all $j \in \{1, \ldots m\}$. Let $U_j$ be an $(\alpha_j, a, \epsilon)$-block-encoding of $A_j$, and $L$ be a unitary acting on $b$ qubits (with $m \leq 2^b-1$) such that $L\ket{0} = \sum_{j=0}^{2^b-1}\sqrt{\alpha_jy_j/s}\ket{j}$, where $s = \sum_{j=1}^my_j\alpha_j$. Then a $(\sum_{j}y_j\alpha_j, a+b, \sum_{j}y_j\alpha_j\epsilon)$-block-encoding of $\sum_{j=1}^my_jA_j$ can be implemented with a single use of
  \begin{align}
      \mathrm{select}(U) \coloneqq \sum_{j=0}^{m-1}\ketbra{j}{j}\otimes U_j + ((I - \sum_{j=0}^{m-1}\ketbra{j}{j})\otimes I _{\bbc^{2^a}}\otimes I_{\bbc^{2^{n}}})
  \end{align} plus twice the cost for implementing $L$. 
\end{lemma}

\subsection{Gaussian quadrature}
\label{sec:def_gaussian}

We are using the Gauss-Legendre quadrature method in this work. The interpolation nodes for Gauss-Legendre quadrature are $-1\leq x_1\leq \cdots \leq x_q\leq 1$ being the roots of $q$-th degree Legendre polynomial, and $w_1,\ldots ,w_q$ are their corresponding weights. 
Then the weight for the $j$-th root of the Gaussian quadrature can be expressed as
\begin{equation}
\label{definition_weights}
   w_j = \frac{2}{(1-x_j^2)P'_n(x_j)^2}.
\end{equation}
Here $P'_n$ is the derivative of Legendre function.
Both weights and roots can be pre-computed and are independent of the integrand. Let $\{{P}_i(x) \}$ be the Legendre polynomials.
By scaling
   $\hat{P}\coloneqq P\left(\frac{2x}{t} - 1\right)$,
we have that $\{\hat s_i = \frac{2}{x_i} - 1\}^q_{i=1}$ are the roots of the scaled $q$-th degree polynomial $\hat{{P}_q}$. 
For one dimensional integral approximation, that is in approximating $\int_0^t f(x)\mathrm{d}x$ using $\sum_{j=1}^qf(\hat s_j)w_j$, the error for the Gaussian quadrature approximation is
\begin{equation}
   E_q[f] = \int^t_0 f(x) \dd x - \sum^q_{j=1} f(\hat{s}_j)w_j = \frac{f^{(2q)}(\xi)}{(2q)!} \int^t_0 \pi_q(x)^2 \dd x,
\end{equation}
for some $\xi \in[0,t]$, where
    $\pi_q(x) \coloneqq (x-\hat s_1)(x-\hat s_2) \cdots (x- \hat s_q)$.
To estimate the integral term in the error, we notice that
\begin{align}
    \pi_q(x) = \frac{t^q q!}{2^q(2q-1)!!} P_q\left(\frac{2x}{t} - 1\right).
\end{align}
This gives an explicit bound
\begin{equation}
   |E_q[f]| = \frac{|f^{(2q)}|(\xi)t^{2q+1}(q!)^2}{(2q)!2^{2q}(2q-1)!!(2q+1)!!}\leq \frac{|f^{(2q)}|(\xi)t^{2q+1}q}{(2q)!2^{4q-1}}.
\end{equation}

In our algorithm, we need to approximate a multiple integral where each level is a scaled Gaussian quadrature. We present the details in \cref{sec:approximate_gaussian}.
In the following, we present several useful lemmas related to Gaussian quadrature weights and roots, whose proofs are postponed to \cref{sec:proof-quadrature}.

\begin{restatable}{lemma}{rootspacing}
\label{lemma:spacing}
    Suppose $\{x_1, x_2, \dots,x_q\}$ are quadrature points for Gauss-Legendre quadrature over interval $[-1,1]$, then it holds that
    \begin{align}
    \label{eq3}
        |x_{j + 1} - x_j| = \frac{\pi}{q} \sqrt{1-x_j^2} + O(1/q^2),
    \end{align}
    for all $j \in \{1,2,\dots,q-1\}$.
\end{restatable}

\begin{restatable}{lemma}{weightbound}
\label{lemma:weight}
    Suppose $\{x_1, x_2, \dots,x_q\}$ are quadrature points for Gauss-Legendre quadrature over interval $[-1,1]$, and $\{w_1, w_2, \dots,w_q\}$ are corresponding weights, then it holds that
    \begin{align}
    \label{eq4}
        w_j = \frac{\pi}{q} \sqrt{1-x_j^2} + O(1/q^2),
    \end{align}
    for all $j \in \{1,2,\dots,q\}$.
\end{restatable}

The following Lemma gives a way to efficiently estimate the partial sum $\sum_{j=a}^b f(\hat{x}_j) w_j$.

\begin{restatable}{lemma}{partialsum}
\label{lemma:partial-sum}
Suppose $\{x_1, x_2, \dots,x_q\}$ are quadrature points for Gauss-Legendre quadrature over interval $[-1,1]$, and $\{w_1, w_2, \dots,w_q\}$ are their corresponding weights, $q$ is the total number of quadrature points, and the first order derivative of $f(x)$ is bounded by
$|f'(x)| \leq l$, then it holds that
\begin{align}
     \left|\sum_{j=a}^b f(\hat{x}_j) w_j - \int_{\hat x_a}^{\hat x_b} f(x) \mathrm{d}x \right| \leq O(l(b-a)/q^2).
\end{align}
\end{restatable}

\section{Simpler discrete-clock construction}
\label{sec:clock}

Our approach redefines the target state \( |\psi(t)\rangle \) as part of an open system coupled with an external ``clock state'' \( |t\rangle \), producing the composite state \( |t\rangle \otimes |\psi(t)\rangle \) upon evolution. To achieve this, we design a Hamiltonian for the coupled system with two components: a clock Hamiltonian that governs the evolution of the clock state and a system Hamiltonian that simulates time-dependent dynamics on the system state based on the clock's values.

To define the clock Hamiltonian \( H_{\mathrm{clk}} \), we start with the informal notion of a continuous clock space, where the evolution of the clock Hamiltonian for time $\tau$ performs the operation
\begin{align}
\label{clock_def0}
    |t\rangle \mapsto |t + \tau\rangle.
\end{align}
This operator resembles the increment operator \( U_+ \), defined as $U_+ |t\rangle = |t+1\rangle$.
This similarity suggests defining \( H_{\mathrm{clk}} \) as $H_{\mathrm{clk}} \sim \log(U_+)$. To make the increment operator unitary, we also make $U_+$ satisfy $U_+ |M-1\rangle = |0\rangle$, here $M$ is the total number of discrete ladder state. 

Let $M$ be the dimension of our discretized clock space and $\delta$ be the time step size satisfying $\delta = T/M$
We formally define the clock Hamiltonian $H_{\mathrm{clk}}$ as
    \begin{align}
        H_{\mathrm{clk}} \coloneqq H_a \otimes I,
    \end{align}
    where $H_a$ satisfies
    \begin{align}
      -i H_a \delta = \log U_+ = Q_{\mathbb{Z}_M} \sum_{x\in \mathbb{Z}_M} (-i2\pi x/M) |x\rangle \langle x| Q_{\mathbb{Z}_M}^{\dagger},
    \end{align}
    where $Q_{\mathbb{Z}_M}$ denotes the quantum Fourier transform on $\mathbb{Z}_M$. In the following, when the domain of the quantum Fourier transform is clear from the context, we may drop the subscript and simply use $Q$ for a concise presentation.
    The system Hamiltonian $H_{\mathrm{sys}}$ is defined as
    \begin{align}
        H_{\mathrm{sys}} \coloneqq \sum_{n=0}^{M -1} |n \rangle \langle n| \otimes H(n\delta),
    \end{align}
    where $H(\cdot)$ is the time-dependent Hamiltonian we want to simulate. The following lemma characterizes the connection between the evolution of $H_{\mathrm{clk}}+H_{\mathrm{sys}}$ 

\begin{lemma}\label{lm1}
  If suffices to choose $M$ to be
    \begin{align}
        M = O\left(\frac{T^4}{\epsilon^2} \max_s \|\dot{H}(s)\| \right).
    \end{align}
    so that
    \begin{align}
        \left\|e^{-i (H_{\mathrm{clk}} + H_{\mathrm{sys}} )T}|0\rangle |\psi\rangle - I \otimes \mathcal{T}\left( e^{-i\int_0^T H(t) \mathrm{d}t}\right) |T\rangle |\psi\rangle\right\| \leq \epsilon,
    \end{align}
    for any state $\ket{\psi}$.
\end{lemma}
\begin{proof}
    
    For simplicity, we use 
    \begin{align}
        U(T,0) = \mathcal{T}\left( e^{-i\int_0^T H(t) \mathrm{d}t}\right) 
    \end{align}
    to denote the evolution operator.
By the definition of integral and Riemann sum, one has
    \begin{align}
        \|(e^{-iH_{\mathrm{clk}}T/M} e^{-iH_{\mathrm{sys}}T/M})^M |0\rangle |\psi(0)\rangle - |T\rangle U(T,0) |\psi(T)\rangle\| \leq \frac{T^2}{2M} \max_s \|\dot{H}(s)\|.
    \end{align}
To make the above error within $\epsilon$, it suffices to choose
\begin{align}
\label{first_choice_of_N}
    M = O\left(\frac{T^2}{\epsilon}\max_s \|\dot{H}(s)\|\right)
\end{align}
Now we need to bound the error for
    \begin{align}
        \|e^{-i(H_{\mathrm{clk}} + H_{\mathrm{sys}}) T} |0\rangle |\psi(0)\rangle - (e^{-iH_{\mathrm{clk}}T/M} e^{-iH_{\mathrm{sys}}T/M})^M |0\rangle |\psi(0)\rangle\| \leq \epsilon_0
    \end{align}
To bound $\epsilon_0$, we first consider a shorter time period $\delta = T/M$ and the corresponding error for this shorter period $\epsilon_1$. During such time period $\delta$, we have
\begin{align}
    e^{-i(H_{\mathrm{clk}} + H_{\mathrm{sys}})\delta} \approx e^{-iH_{\mathrm{clk}} \delta} e^{-i H_{\mathrm{sys}} \delta} \left(1 - \frac{\delta^2}{2}[H_{\mathrm{clk}},H_{\mathrm{sys}}] + O(\delta^3)\right).
\end{align}
One may want to bound the spectral error of the above operator as
\begin{align}
    \epsilon_1 = O\left(\frac{T^2}{2 M^2} \|[H_{\mathrm{clk}},H_{\mathrm{sys}}]\|\right).
\end{align}
However, when we were acting on states with a single clock state $|t\rangle |\psi\rangle$, we can bound the error better. Consider we are acting on $|t_0\rangle |\psi\rangle$, the error term will be
\begin{align}
    \|[H_{\mathrm{clk}},H_{\mathrm{sys}}] |t_0\rangle |\psi\rangle\| 
    &= \left\|\frac{1}{\delta N} \sum_{t}^{M-1} \sum_{m=0}^{M-1} \omega_r^{m(t-t_0)} \cdot \frac{2\pi m}{M} |t\rangle \otimes (H(t_0\delta) - H(t\delta))|\psi\rangle \right\|\\
    &= O\left( \left\|\sum_{t=0}^{M-1} |t\rangle \otimes \left(\frac{H(t_0 \delta) - H(t \delta)}{t_0 \delta - t\delta}\right)|\psi\rangle\right\|\right).
\end{align}
This can be bounded by $O(\sqrt{M} \max_s \|\dot{H}(s)\|)$. Therefore, $\epsilon_1$ can be bounded by 
\begin{align}
    \epsilon_1 = \|e^{-iH_{\mathrm{clk}} \delta} e^{-i H_{\mathrm{sys}} \delta} \cdot \frac{\delta^2}{2}[H_{\mathrm{clk}},H_{\mathrm{sys}}] |t_0\rangle |\psi\rangle\| =  O\left(\frac{T^2}{2 M^{3/2}} \max_s \|\dot{H}(s)\|\right).
\end{align}
Combing $M$ terms together, the total error $\epsilon_0$ is then 
\begin{align}
    \epsilon_0 = O\left(\frac{T^2}{2 M^{1/2}} \max_s \|\dot{H}(s)\|\right).
\end{align}
To make $\epsilon_0 \leq \epsilon$, it suffices to choose the total number of time steps $M$ to be
\begin{align}
\label{eq:order_of_M}
    M = O\left(\frac{T^4}{\epsilon^2} \max_s \|\dot{H}(s)\| \right).
\end{align}
Combining with \cref{first_choice_of_N}, it suffices to choose \cref{eq:order_of_M}
\end{proof}

By choosing $M$ as \cref{eq:order_of_M}, the norm of $H_{\mathrm{clk}}$ is bounded by
\begin{align}
    \|H_{\mathrm{clk}}\| = \|H_a\| = \frac{1}{\delta} = \frac{M}{T} = O\left(\frac{T^3}{\epsilon^2} \max_s \|\dot{H}(s)\| \right).
\end{align}
After applying quantum Fourier transform $Q$ on both sides of $H_{\mathrm{clk}}$, the resulting operating, denoted by $D$ is 
\begin{align}
\label{eq:def_operator_D}
    D \coloneqq (Q^{\dagger} \otimes I)H_{\mathrm{clk}} (Q\otimes I) = \frac{M}{T}\sum_{x\in \mathbb{Z}_M} (2\pi x/M) |x\rangle \langle x| \otimes I.
\end{align}
Note that $D$ is a diagonal matrix whose exponential \( e^{-iDT} \) can be efficiently implemented. The norm of $D$ is upper bounded by
\begin{align}
  \label{eq:D_norm}
  \|D\| = \|H_{\mathrm{clk}}\| = O\left(\frac{T^3}{\epsilon^2} \dot{H}_{\max} \right).
\end{align}
Since
\begin{align}
  QQ^{\dagger} e^{-i(H_{\mathrm{sys}} + {H}_{\mathrm{clk}}) T}  Q Q^{\dagger} 
     = Q e^{-i(Q^{\dagger} H_{\mathrm{sys}} Q + Q^{\dagger} {H}_{\mathrm{clk}} Q)T} Q^{\dagger} = Q e^{-i(Q^{\dagger}H_{\mathrm{sys}} Q + D)T} Q^{\dagger},
\end{align}
the system Hamiltonian in the Fourier basis, denoted by $B$, is given as
\begin{align}
\label{eq:operator_FHF}
    B \coloneqq (Q^{\dagger} \otimes I )H_{\mathrm{sys}} (Q \otimes I).
\end{align}
Now, it follows that
\begin{align}
    (Q \otimes I) e^{-i(B + D)T} (Q^{\dagger} \otimes I) = e^{-i(H_{\mathrm{sys}} + {H}_{\mathrm{clk}}) T}.
\end{align}
 Since the Fourier transform operation is unitary, one has
\begin{align}
    \|B\|_{\max} = \|H_{\mathrm{sys}}\|_{\max} = H_{\max},
\end{align}
For the rest of this paper, we consider how to efficiently implement \( e^{-i(D + B)T} \).

\section{Simulation algorithm}

\subsection{Treatment by Duhamel's principle}

Next, let us see how to efficiently implement $e^{-i(D+B)t}$ \footnote{Note that this $t$ here in $e^{-i(D + B) t}$ should not be confused with $T$ we defined previously, which denotes the total simulation time. This $t$ is in $[0,T]$.} with Duhamel's principle and Gaussian quadrature. 
\begin{restatable}{lemma}{duhamel}
    Given matrices $D$ and $B$ of the same dimension, for all $t>0$, it holds that
    \begin{align}
        e^{-i(D + B)t} = & e^{-iDt} + (-i) \int_0^t e^{-iD(t-s)} B e^{-i(D + B) s} \ \mathrm{d}s.
    \end{align}
\end{restatable}
To prove the above lemma, it suffices to check that both sides solve the same first-order linear system. For completeness, we present the proof in \cref{sec:proof-technical}.
Then we recursively substitute the expression of $e^{-i (D + B)s}$ in the integrand to obtain the following 
\begin{align}
    \label{eq:2variable_expansion}
    e^{-i(D + B)t} = & e^{-iDt} + (-i) \int_0^t e^{-iD(t-s_1)} B e^{-iD s_1} \, \mathrm{d}s_1 \\
    &+ (-i^2) \int_0^t \int_0^{s_1} e^{-iD(t-s_1)} B e^{-iD(s_1 - s_2)} B e^{-iD s_2} \, \mathrm{d}s_2 \, \mathrm{d}s_1 \\
    &+ (-i)^3 \int_0^t \int_0^{s_1} \int_0^{s_2} e^{-iD(t-s_1)} B e^{-iD(s_1 - s_2)} B e^{-iD(s_2 - s_3)} B e^{-iD s_3} \, \mathrm{d}s_3 \, \mathrm{d}s_2 \, \mathrm{d}s_1 \\
    &+ \cdots \\
    &+ (-i)^K \int_0^t \int_0^{s_1}  \cdots \int_0^{s_{K-1}} e^{-iD(t-s_1)} B e^{-iD(s_1 - s_2)} B e^{-iD(s_2 - s_3)} \cdots B e^{-iD s_K} \, \mathrm{d}s_K \cdots \mathrm{d}s_1 \\
    &+ O\left(\frac{t^K}{K!} \|B\|^K\right),
\end{align}
where we have used
\begin{align}
    \|e^{-iD(t-s_1)} B e^{-iD(s_1 - s_2)} B e^{-iD(s_2 - s_3)} \cdots B e^{-iD s_K} \| = \|B\|^K,
\end{align}
and 
\begin{align}
    \int_0^t \int_0^{s_1} \int_0^{s_2} \cdots \int_0^{s_{K-1}} \, \mathrm{d}s_K \cdots \mathrm{d}s_3 \, \mathrm{d}s_2 \, \mathrm{d}s_1 \leq O\left(\frac{t^K}{K!}\right).
\end{align}
If we define
\begin{align}
\label{def:F_k}
    F_k(s_k, \dots, s_1) = e^{i D s_k} B e^{-i D(s_k - s_{k-1})} B \cdots B e^{-i D s_1},
\end{align}
then \cref{eq:2variable_expansion} can be rewritten as
\begin{align}
\label{eq:expansion_of_du}
    e^{-i(D + B)t} = e^{-iDt} \left(I + \sum_{k=1}^K (-i)^k \int_{0 \leq s_1 \leq \cdots \leq s_k \leq t} F_k(s_k, \dots, s_1) \, \mathrm{d}s_1 \cdots \mathrm{d}s_k\right) + O\left(\frac{t^K}{K!}\|B\|^K\right) .
\end{align}

As a side note, our approach is significantly different from the truncated Dyson series approach, which is outlined in the following approximation.
\begin{align}
\label{eq:dyson}
    \mathcal{T} \exp \left(-i \int_0^t H(\tau) \mathrm{d} \tau\right) = I + \sum_{k=1}^K (-i)^k \int_{0 \leq s_1 \leq \cdots \leq s_k \leq t}H(s_k)\cdots H(s_1)\mathrm{d}s_1 \cdots \mathrm{d}s_n + O\left(\frac{t^K}{K!}{(H_{\max})}^K\right),
\end{align}
where the integrand contains time-dependent Hamiltonians. On the other hand,
the nontrivial terms in \cref{eq:expansion_of_du} are time‑independent. The only time‑dependent contribution is the diagonal Hamiltonian $e^{-iDs}$, which can be fast-forwarded. 

\subsection{Approximating multiple integrals using scaled Gaussian quadrature}
\label{sec:approximate_gaussian}

To compute multiple integrals, we need to use scaled Gaussian quadrature multiple times. For an interval $[0,s_k]$ with $s_k\leq t$, we use scaled quadrature points and weights: $s_k x_1/t,\ldots,s_k x_q/t$, and $s_k w_1/t,\ldots,s_k w_q/t$. Then, $\int^{s_k}_0 f(x)\dd x$ can be approximated by the scaled quadrature points and weights:
\begin{equation}
   \int^{s_k}_0 f(x)\mathrm{d} x \approx \sum^q_{j=1}f\left(\frac{s_k x_j}{t}\right) \frac{s_k w_j}{t}.
\end{equation}
For each $j\in [q]$, define the functions $u_k$ and $v_k$ as
\begin{equation}
   u_j(x) \coloneqq x\hat{s}_j/t, \quad v_j(x) \coloneqq xw_j/t.
\end{equation}
We also define
\begin{equation}
\label{eq:def_multi_x}
   \hat{x}_{j_k}\coloneqq \hat{s}_{j_k},\text{   and   }\hat{x}_{(j_k,\ldots,j_{k-l})} \coloneqq u_{j_{k-l}} \circ \cdots \circ u_{j_{k-1}}(\hat{s}_{j_k}) \text{ for all }1\leq l\leq k-1,
\end{equation}
\begin{equation}
\label{eq:def_multi_w}
   \hat{w}_{(j_k)}\coloneqq w_{j_k}, \text{ and } \hat{w}_{(j_k,\ldots,j_{k-l})}\coloneqq v_{j_{k-l}}(\hat{x}_{(j_k,\ldots,j_{k-l+1})}) \text{ for all }1\leq l\leq k-1.
\end{equation}
Here $u_{j_1} \circ u_{j_2}(x) \coloneqq x \hat s_{j_1} \hat s_{j_2} /t$.
With these notations, we can approximate multiple-dimensional integrals
\begin{align}
    \int_{0\leq s_1 \leq \cdots \leq s_k \leq t} F(s_k,\dots,s_1) \mathrm{d}s_1 \cdots \mathrm{d}s_k
\end{align}
with 
\begin{align}
    \sum_{j_1 = 1}^q \cdots \sum_{j_k = 1}^q F (\hat{x}_{(j_k)}, \dots, \hat{x}_{(j_k, \dots, j_1)}) \hat{w}_{(j_k)} \cdots \hat{w}_{(j_k, \dots, j_1)},
\end{align}
where $F$ is a function with $k$ variables.

Now, for \cref{eq:expansion_of_du}, each term in this series can be approximated using Gaussian quadrature. Specifically, we approximate it by multiple-dimensional Gaussian quadrature
\begin{align}
\label{final_simulation_sum}
    e^{-iDt} \left(I + \sum_{k=1}^K (-i)^k \sum_{j_1 = 1}^q \cdots \sum_{j_k = 1}^q F_k (\hat{x}_{(j_k)}, \dots, \hat{x}_{(j_k, \dots, j_1)}) \hat{w}_{(j_k)} \cdots \hat{w}_{(j_k, \dots, j_1)}\right).
\end{align}
Since \( e^{-iDt} \) is straightforward to implement, we focus on implementing the series:
\begin{align}
    \label{sum_series}
    I + \sum_{k=1}^K (-i)^k \sum_{j_1 = 1}^q \cdots \sum_{j_k = 1}^q F_k (\hat{x}_{(j_k)}, \dots, \hat{x}_{(j_k, \dots, j_1)}) \hat{w}_{(j_k)} \cdots \hat{w}_{(j_k, \dots, j_1)}.
\end{align}
Let us first bound the error between the Gaussian quadrature and the integral.
\begin{restatable}{lemma}{fkbound}
\label{lemma_fkbound}
Suppose $F_k$ is defined as in \cref{def:F_k}.
Let $\hat x_{(j_1,\dots,j_k)}$ and $\hat w_{(j_1,\dots,j_k)}$ for $j_{\ell}\in [n]$ be the scaled Gaussian quadrature weights and roots as we defined \cref{eq:def_multi_x,eq:def_multi_w}. It holds that
    \begin{align} 
    &\left\|\int_{0 \leq s_1 \leq \cdots \leq s_k \leq t} F_k(s_k, \dots, s_1) \, \mathrm{d}s_1 \cdots \mathrm{d}s_k - \sum_{j_1 = 1}^q \cdots \sum_{j_k = 1}^q F_k (\hat{x}_{(j_k)}, \dots, \hat{x}_{(j_k, \dots, j_1)}) \hat{w}_{(j_k)} \cdots \hat{w}_{(j_k, \dots, j_1)} \right\|\\
    & \leq O\left(\frac{(2t)^{k-1}}{(k-1)!} \frac{\|D\|^{2q}\|B\|^{k} t^{2q}q}{(2q)!}\right). 
\end{align}
\end{restatable}
The proof is postponed to \cref{sec:proof-technical}.
By \cref{lemma_fkbound}, for all $K$ terms, the total error is
\begin{align}
    &\left\|\sum_{k=0}^K\int_{0\leq s_1 \leq \cdots \leq  s_k\leq t}F_k(s_k)\cdots F_k(s_1)\mathrm{d}s_1\cdots\mathrm{d}s_k - 
    \sum_{k=0}^K\sum^q_{j_1 = 1}\cdots \sum^q_{j_k = 1}F_k(\hat{x}_{(j_k)})\cdots F_k(\hat{x}_{(j_k,\dots,j_1)})\hat{w}_{(j_k)}\cdots \hat{w}_{(j_k,\dots , j_1)} \right\|\\
    &\leq \sum_{k=0}^K \left\|\int_{0\leq s_1 \leq \cdots \leq  s_k\leq t}F_k(s_k)\cdots F_k(s_1)\mathrm{d}s_1\cdots\mathrm{d}s_k - \sum^q_{j_1 = 1}\cdots \sum^q_{j_k = 1}F_k(\hat{x}_{(j_k)})\cdots F_k(\hat{x}_{(j_k,\dots,j_1)})\hat{w}_{(j_k)}\cdots \hat{w}_{(j_k,\dots , j_1)} \right\| \\
    &= O\left(\sum_{k=0}^K \frac{(2t)^{k-1}}{(k-1)!} \frac{\|D\|^{2q}\|B\|^{k} t^{2q}q}{(2q)!} \right) = O\left(\frac{\|D\|^{2q} t^{2q + 1} q}{(2q)!} \left(e^{4t\|B\|}\right)\right).
\end{align}
Therefore, the total error of using \cref{final_simulation_sum} to simulate the time-dependent evolution operator is
\begin{align}
    O\left(\frac{(t\|B\|)^{K+1}}{(K+1)!} + \frac{\|D\|^{2q} t^{2q + 1} q}{(2q)!} \left(e^{4t\|B\|}\right) \right),
\end{align}
where the first term comes from Duhamel's principle and the second term comes from Gaussian quadrature. 
Now let us choose $q$ and $K$ such that the above error is within $\epsilon$.
We first consider time period $[0,t]$, such that $t\alpha = \Theta(1)$. This ensures $t\|B\| = \Theta(1)$. To make the first term within error $\epsilon$, it suffices to choose
    $K = O\left(\frac{\log(1/\epsilon)}{\log\log(1/\epsilon)}\right)$.
    Given that \( \|D\| \) is expected to be large, the second term dominates. By \cref{eq:D_norm}, we need  
    $\frac{T^{8q + 1}q (\dot H_{\max})^{2q}}{(2q)!  {\epsilon}^{4q}} \leq \epsilon$.
Applying Stirling’s approximation,
we have $\log q\left(1-\frac{1}{2q}\right) \geq \log(T^4\dot H_{\max}/\epsilon^2)$,
requiring that
    $q \geq \Omega(T^4\dot H_{\max}/\epsilon^2)$,
which represents the number of quadrature points. For simulating time period $T$, there are $O(T\alpha)$ slices in total (see our analysis later explaining why we are choosing this number of time slices). The error for each slice should be within
    $O\left(\frac{\epsilon}{T\alpha}\right)$.
Then the value of $q$ and $K$ should be at least
\begin{align}
    \label{eq:q-k-bound}
    q \geq \Omega\left(\frac{T^6 \alpha^2 \dot H_{\max}}{\epsilon^2}\right), \quad 
    K = \Theta\left(\frac{\log(T\alpha/\epsilon)}{\log \log (T\alpha/\epsilon)}\right).
\end{align}

\subsection{Algorithmic implementation}
Now we use \cref{lemma:sum-to-be} and Gaussian quadrature to approximate \cref{eq:2variable_expansion}.

Let us construct the block-encoding for $B$ and $e^{-iD t}$ for all $t \in [0,T]$. One can construct a block-encoding of $B$ with normalizing factor $\|B\|_{\max}$ by controlled rotation and swap operator, which requires $O(1)$ queries to $H(t)$ and additional $O(n)$ gates, where $n$ is the number of qubits to in the Hamiltonian system. Since $D$, previously defined as in \cref{eq:def_operator_D},
is diagonal, we can construct $e^{-i D t}$ by controlled rotations. Therefore, we can efficiently construct $e^{-iD t}$ with $O(1)$ gates. 
To implement \cref{sum_series} with quantum circuits with \cref{lemma:sum-to-be}, we first construct 
\begin{align}
\label{oracle_AI}
    \sum_t |t\rangle\langle t| \otimes (e^{iD t} U_B e^{-iD t}) = \sum_t |t\rangle \langle t| \otimes V(t),
\end{align}
as the controlled unitary operator,
with $U_B$ being the block-encoding of $B$. Note that we have defined $V \coloneqq e^{iDt} U_B e^{-iDt}$ to simplify the presentation. 
We have $\mathrm{HAM_T}$ defined in \cref{oracle_inputblockencoding} being a block-encoding of $H_{\mathrm{sys}}$ with normalizing constant $\alpha$. This gives construction of $U_B$ as $U_B = (Q^{\dagger} \otimes I)(\mathrm{HAM_T})(Q\otimes I)$. 
The normalizing constant of the block-encoding $U_B$ is $\alpha = d H_{\max}$, and it can be implemented using $O(1)$ queries to $\mathrm{HAM_T}$ by \cref{oracle_inputblockencoding}. 
By \cite[Lemma 53]{GSLW19}, we know that $V(\hat{x}_{(j_k)}) V(\hat{x}_{(j_k,j_{k-1})})\cdots V(\hat{x}_{j_k, j_{k-1},\dots,j_1})$ is an $\alpha^k$ -block-encoding of
    $F_k(\hat{x}_{(j_k)}, \hat{x}_{(j_k,j_{k-1})}, \dots, \hat{x}_{j_k, j_{k-1},\dots,j_1})$, here $F_k$ is defined as in \cref{def:F_k}.

Now, we need to implement the $\mathrm{select}(U)$ as in \cref{lemma:sum-to-be}. In our case, the $\mathrm{select}(U)$ is defined as
\begin{equation}
\label{oracle_select_U}
\begin{aligned}
    \mathrm{select}(U) &\coloneqq \sum_{n=1}^K \sum_{j_1 = 1}^q \cdots \sum_{j_k = 1}^q|k, j_1,\ldots,j_k\rangle \langle k, j_1,\ldots,j_k|\otimes V(\hat{x}_{(j_k)}) V(\hat{x}_{(j_k,j_{k-1})})\cdots V(\hat{x}_{j_k, j_{k-1},\dots,j_1}) \\
   &+ ((I- \sum_{n=1}^K \sum_{j_1 = 1}^q \cdots \sum_{j_k = 1}^q|k, j_1,\ldots,j_k\rangle \langle k, j_1,\ldots,j_k|)\otimes I).
\end{aligned}
\end{equation}
This means that for all $k\in [K],$ $j_1,\ldots,j_k \in [q]$, the effect of this operator is
\begin{equation}
\label{select_v_target}
   |k, j_1,\ldots,j_k\rangle |\psi\rangle \mapsto |k, j_1,\ldots,j_k\rangle V(\hat{x}_{(j_k)}) V(\hat{x}_{(j_k,j_{k-1})})\cdots V(\hat{x}_{j_k, j_{k-1},\dots,j_1})|\psi\rangle + |\phi^{\perp}\rangle,
\end{equation}
Where $|\phi^{\perp}\rangle$ is some state that has zero overlap with the first part.

Suppose our initial state is $|k,j_1, \ldots, j_k\rangle |0\rangle$. 
Using standard techniques, we can prepare the state $|k,j_1, \ldots, j_k, \hat{x}_{(j_1)}, \ldots, \hat{x}_{(j_k)}\rangle$. Here $|k\rangle$ takes unary form, that is
   $|k\rangle = |1,\ldots,1,0, \ldots,0\rangle = |1^k 0^{K-k}\rangle$,
where each of the $1$s and $0$s corresponds to each registers $|j_1\rangle, \ldots, |j_k\rangle$. Those after position $n$ are changed to $|1\rangle$ when the first register is in state $|k\rangle$. This means that the final state is
   $|k,j_1, \ldots, j_k, \hat{x}_{(j_1)}, \ldots, \hat{x}_{(j_k)},1^{K-k}\rangle$.
For any input states
   $|k,j_1, \ldots, j_k, \hat{x}_{(j_1)}, \ldots, \hat{x}_{(j_k)},1^{K-k}\rangle$,
we first apply multiplication operator $M_t$ to all adjacent registers $|\hat{x}_{j_l}\rangle $ and $|\hat{x}_{j_{l-1}}\rangle$:
\begin{equation}
\label{oracle_Mt}
   M_t|x\rangle |y\rangle = |\frac{1}{t} x y\rangle |y\rangle,
\end{equation}
where $t$ is the total time we want to simulate. This operator has gate complexity $O(m\log m)$ \cite{HV21} for $m$-bit integer. To make the precision to be below $O(\epsilon)$, it suffices to choose $m$ to be $m = O(\log 1/\epsilon)$. Now we have state
\begin{equation}
   |k,j_1, j_2, \ldots, j_k\rangle |\hat{x}_{(j_k,\ldots,j_2, j_1)}, \ldots, \hat{x}_{(j_k)},1^{K-k}\rangle,
\end{equation}
where $\hat x_{(j_k, j_{k-1}, \dots, j_1)}$ is defined in \cref{eq:def_multi_x}.
Finally, we apply $A_I$ to each control registers $|\hat{x}_{(j_k,\ldots,j_{k-l})}\rangle$ and target state $|\psi\rangle$. Then we can get
\begin{equation}
   |k,j_1,j_2,\ldots,j_k\rangle|\hat{x}_{(j_k,\ldots,j_2, j_1)}, \ldots, \hat{x}_{(j_k)},1^{K-k}\rangle V(\hat{x}_{(j_k)})\cdots V(\hat{x}_{j_k, j_{k-1},\ldots,j_1})|\psi\rangle + |\phi^{\perp}\rangle.
\end{equation}
Then we uncompute the ancillary qubits to get
\begin{equation}
   |k, j_1,j_2,\ldots,j_k\rangle|0\rangle V(\hat{x}_{(j_k)})\cdots V(\hat{x}_{j_k, j_{k-1},\ldots,j_1})|\psi\rangle + |\phi^{\perp}\rangle,
\end{equation}
which corresponds to our initial target \cref{select_v_target}.
So, to prepare \cref{select_v_target}, we need $O(n)$ queries to $V$, where $n$ is the number of qubits in the Hamiltonian system. The number of qubits required to implement this algorithm depends on the precision. To make the precision to be within $O(\epsilon)$, each ancillary register $|\hat{x}_{(j_k, \ldots, j_{k-l})}\rangle$ requires $O(\log 1/\epsilon)$ ancillary qubits, which brings our total ancillary qubits for the linear combination of block-encodings to be $O(\frac{\log^2(1/\epsilon)}{\log\log (1/\epsilon)})$. For total evolution time $T$, it is 
\begin{align}
    O\left( \frac{\log^2(T\alpha/\epsilon)}{\log\log (T\alpha/\epsilon)} \right).
\end{align}
If we see time register as ancillary qubits, the total number of ancillary qubits is then
\begin{align}
    O\left( \frac{\log^2(T\alpha/\epsilon)}{\log\log (T\alpha/\epsilon)}  + \log \left(M\right)\right) = O\left( \frac{\log^2(T\alpha/\epsilon)}{\log\log (T\alpha/\epsilon)}  + \log \left(\frac{T}{\epsilon}\dot H_{\max}\right)\right).
\end{align}
We also need $O(K)$ implementations of $M_t$ as in \cref{oracle_Mt}. Since the gate complexity is dominated by $M_t$, the total gate cost of preparing $\mathrm{select}(U)$ as in \cref{oracle_select_U} is $O(\log^2 (1/\epsilon))$.

Next, we need to construct the operator $L$ in \cref{lemma:sum-to-be}, such that it prepares the following state
\begin{equation}
\label{operator_B}
   L|0\rangle = \frac{1}{\sqrt{s}}  \sum_{k=1}^K \sum_{j_1 = 1}^q \cdots \sum_{j_k = 1}^q \sqrt{\alpha^k \hat{w}_{(j_k)}\cdots \hat{w}_{(j_k,\ldots j_1)}}|k, j_1,\ldots,j_k\rangle,
\end{equation}
where $s$ is the normalizing constant, and $\alpha$ is the block-encoding constant of \cref{oracle_AI}. Note that
\begin{equation}
   \hat{x}_{(j_k,j_{k-1},\ldots,j_{k-l})} = \frac{1}{t^l} \hat{x}_{(j_k)} \hat{x}_{(j_{k-1})} \cdots \hat{x}_{(j_{k-l})}.
\end{equation}
We further have
\begin{equation}
   \hat{w}_{(j_k,j_{k-1},\ldots,j_{k-l})} = \frac{1}{t^l} \hat{x}_{(j_k)} \hat{x}_{(j_{k-1})} \cdots \hat{x}_{(j_{k-l+1})} w_{(j_{k-l})},
\end{equation}
so
\begin{equation}
\begin{aligned}
   \hat{w}_{(j_k)}\cdots \hat{w}_{(j_k,\ldots j_1)}
   &= \frac{1}{t^{1+2+\cdots +(k-1)}} \hat{x}_{(j_k)}^{k-1} \hat{x}_{(j_{k-1})}^{k-2} \cdots \hat{x}_{(j_2)} w_{j_k} w_{j_{k-1}} \cdots w_{j_1}\\
   &= \frac{1}{t^{k(k-1)/2}} \hat{x}_{(j_k)}^{k-1} \hat{x}_{(j_{k-1})}^{k-2} \cdots \hat{x}_{(j_2)} w_{j_k} w_{j_{k-1}} \cdots w_{j_1}.
\end{aligned}
\end{equation}
Then the state in \cref{operator_B} can be written as
\begin{equation}
\label{oracle_L}
   L|0\rangle = \frac{1}{\sqrt{s}}  \sum_{n=1}^K \sum_{j_1 = 1}^q \cdots \sum_{j_k = 1}^q \frac{1}{\sqrt{t^{k(k-1)/2}}} \sqrt{\alpha^k \hat{x}_{(j_k)}^{k-1} \hat{x}_{(j_{k-1})}^{k-2} \cdots \hat{x}_{(j_2)} w_{j_k} w_{j_{k-1}} \cdots w_{j_1}}|k, j_1,\dots,j_k\rangle.
\end{equation}
Since this is a superposition of $O(q^K)$ different quantum states, a direct implementation of this would require a gate cost of $O(q^K)$, which is not efficient with the choices of $K$ and $q$. 
To address this, we further rewrite the right-hand-side of \cref{oracle_L} as
\begin{equation}
\label{eq:target_state}
   \frac{1}{\sqrt{s}}  \sum_{k=1}^K \sqrt{\frac{\alpha^k}{t^{k(k-1)/2}}} |k\rangle \sum_{j_1 = 1}^q \sqrt{w_{j_1}}|j_1\rangle \sum_{j_2 = 1}^q\sqrt{\hat{x}_{j_2} w_{j_2}}|j_2\rangle \cdots \sum_{j_k = 1}^q  \sqrt{\hat{x}_{(j_k)}^{k-1} w_{j_k}} |j_k\rangle,
\end{equation}
by our previous analysis in \cref{eq:q-k-bound}, the value of $q$ and $K$ should be at least
\begin{align}
    \label{eq:q-bound}
    q \geq \Omega\left(\frac{T^4 \dot H_{\max}}{\epsilon^2}\right), \quad 
    K = \Theta\left(\frac{\log(1/\epsilon)}{\log \log (1/\epsilon)}\right).
\end{align}
To prepare \cref{eq:target_state}, we can prepare states $\sum_{k=1}^K \sqrt{\frac{\alpha^k}{t^{k(k-1)/2}}} |k\rangle, \sum_{j_1 = 1}^q \sqrt{w_{j_1}}|j_1\rangle, \dots, \sum_{j_k = 1}^q  \sqrt{\hat{x}_{(j_k)}^{k-1} w_{j_k}} |j_k\rangle$ separately, and tensor them together. For state $\sum_{k=1}^K \sqrt{\frac{\alpha^k}{t^{k(k-1)/2}}} |n\rangle$, we can use method introduced in \cite{SBM06} to prepare with gate complexity $O(K)$. For other states 
    $\sum_{j_1 = 1}^q \sqrt{w_{j_1}}|j_1\rangle$, $\dots$, $\sum_{j_k = 1}^q  \sqrt{\hat{x}_{(j_k)}^{k-1} w_{j_k}} |j_k\rangle$,
they are independent of the input Hamiltonian. We can further speed up the gate complexity for state preparation to $O(\log q)$ by the following lemma together with the Grover-Rudolph method~\cite{GR02}. 
\begin{restatable}{lemma}{stateprep}
    Suppose $P_q(x)$ is Legendre polynomial of degree $q$, $\{x_i\}_{i=1}^q$ and $\{w_i\}_{i=1}^q$ are corresponding Gauss-Legendre quadrature roots and weights. One can prepare state $|\psi\rangle$ using $O(\log 1/\epsilon)$ 1- and 2-qubit gates such that
    \begin{equation}
    \label{target_state}
        \left\||\psi\rangle - \frac{1}{\sqrt{\sum_{j = 1}^q  x_{j}^\ell w_j}}\sum_{j = 1}^q  \sqrt{x_{j}^\ell w_j} |j\rangle \right\| \leq O(\epsilon),
    \end{equation}
    by choosing $q = \Omega\left(\frac{1}{\epsilon^2}\log \frac{1}{\epsilon}\right)$, where $\ell = O(\log 1/\epsilon)$.
\end{restatable}
The proof is postponed to \cref{sec:proof-technical}.
To guarantee that the final state in \cref{operator_B} is approximated to within an error~$\epsilon$, it suffices to prepare each subregister in the state
\begin{equation}\label{eq:subregister-state}
  \sum_{j=1}^q \sqrt{\hat{x}_{(j)}^{\,l}\,w_{j}}\;\bigl|j\bigr\rangle
\end{equation}
with precision
\begin{equation}\label{eq:error-bound}
  O\left(\frac{\epsilon}{k}\right)
  \;=\;
  O\left(\frac{\epsilon}{\log(1/\epsilon)}\right),
\end{equation}
where
\begin{equation}\label{eq:k-definition}
  k = \Theta\left(\frac{\log(1/\epsilon)}{\log\log(1/\epsilon)}\right).
\end{equation}
By enforcing that each amplitude‐encoding subroutine incurs at most $O(\epsilon/k)$ error, the total error in the overall state remains bounded by~$\epsilon$.

By our result in state preparation, this requires our $q$ to be 
    $q = \Omega\left({\frac{k^2}{\epsilon^2}\log \frac{k}{\epsilon}}\right) = \Omega\left(\frac{1}{\epsilon^2} \log^3 \frac{1}{\epsilon}\right)$.
Comparing this with our previous requirement for $q$ in \cref{eq:q-bound},
to minimize our gate complexity, it suffices to choose 
    $q = O\left(\frac{T^4\dot H_{\max}}{\epsilon^2} \log^3\left(\frac{1}{\epsilon}\right)\right)$.
Then the gate complexity of constructing $L$ as in \cref{operator_B} is 
\begin{align}
    O(K + K\log q) = \left(\log\left(\frac{T\dot H_{\max}}{\epsilon}\right)\cdot\frac{\log(1/\epsilon)}{\log\log(1/\epsilon)}\right).
\end{align}
The remaining part of our gate complexity comes from the system. For every query oracle in \cref{oracle_Ham}, it requires $O(n)$ gates to implement it. 

We are ready to use \cref{lemma:sum-to-be} to implement
an $\alpha_W$-block-encoding of \cref{sum_series}, where $\alpha_W$ is
\begin{align}
   \alpha_W
   &= \sum_{k=1}^K \sum_{j_1 = 1}^q \cdots \sum_{j_k = 1}^q \frac{1}{{t^{k(k-1)/2}}} {\alpha^n \hat{x}_{(j_k)}^{k-1} \hat{x}_{(j_{k-1})}^{k-2} \cdots \hat{x}_{(j_2)} w_{j_k} w_{j_{k-1}} \cdots w_{j_1}}\\
   &= \sum_{k=1}^K \alpha^n \sum_{j_1 = 1}^q \cdots \sum_{j_k = 1}^q \hat{w}_{(j_k)}\cdots \hat{w}_{(j_k,\dots, j_1)} = \sum_{k=1}^K \frac{(\alpha t)^k}{k!} = O(e^{\alpha t}).
\end{align}
To boost the success probability of oblivious amplitude amplification~\cite{BCCKS15}, we split the whole time period into smaller ones as we did previously. Suppose the evolution time is $t$, still we make $t\alpha = \Theta(1)$. Then for each segment, the evolution time is $t = \Theta(1/\alpha)$.
Now then $\alpha_W = O(1)$, the total gate cost is still dominated by $\text{select}(U)$ as in \cref{oracle_select_U} and $L$ as in \cref{oracle_L}. 

Now, suppose the total evolution time is $T$, and there are $O(T/\alpha)$ segments in total. For each segment, the error is required to be within $O(\epsilon/T\alpha)$, and the total gate complexity is then 
\begin{align}
    O\left(\alpha T \left(n + \log\left(\frac{T^2\alpha\dot H_{\max}}{\epsilon}\right) \right) \cdot\frac{\log(T\alpha/\epsilon)}{\log\log(T\alpha/\epsilon)}\right).
\end{align}
In the sparse Hamiltonian input model, it is then
\begin{align}  
    O\left(dTH_{\max}\left(n +
    \log\left(\frac{T^2dH_{\max}\dot H_{\max}}{\epsilon}\right)\cdot\frac{\log(TdH_{\max}/\epsilon)}{\log\log(TdH_{\max}/\epsilon)}\right)\right),
\end{align}
where $d$ is the sparsity of our Hamiltonian.
The query complexity is then
\begin{align}
    O\left(dTH_{\max}\frac{\log(TdH_{\max}/\epsilon)}{\log\log(TdH_{\max}/\epsilon)}\right),
\end{align}
and the total number of ancillary qubits is
\begin{align}
    O\left( \frac{\log^2(TdH_{\max}/\epsilon)}{\log\log (Td H_{\max}/\epsilon)}  + \log \left(\frac{T}{\epsilon}\dot H_{\max}\right)\right).
\end{align}
This proves \cref{thm:main-theorem}.

\clearpage
\bibliographystyle{alpha}
\bibliography{ref}

\appendix
\section{Proofs for Gaussian quadrature}
\label{sec:proof-quadrature}
\rootspacing*
\begin{proof}
  By \cite[Eq.~(3.2)]{Brassthegaussian96}, we have
\begin{align}
    x_j = \cos \left(\frac{\pi}{q + 1/2} (j - \frac{1}{4})\right) + O(1/q^2).
\end{align}
By Taylor expansion, we expand $\cos(y + \delta)$ at $y$ as
\begin{align}
    \cos(y + \delta) 
    &= \cos y + \delta \cos' y  + \frac{1}{2} \delta^2\cos'' y \\
    &= \cos y - \delta \sin y   + O(\delta^2).
\end{align}
If we set $y = \frac{\pi}{q + \frac{1}{2}} (j - \frac{1}{4})$ and $\delta = \frac{\pi}{q + \frac{1}{2}}$, we can approximate $x_{j + 1}$ by
\begin{align}
    x_{j + 1} &= \cos \left(\frac{\pi}{q + \frac{1}{2}} (j + 1 - \frac{1}{4})  \right)  \\
    &= \cos y - \sin y \cdot \delta + O(\delta^2)\\
    &=  \cos \left(\frac{\pi}{q + \frac{1}{2}} (j - \frac{1}{4})\right) - \sin \left(\frac{\pi}{q + \frac{1}{2}} (j - \frac{1}{4})\right) \frac{\pi}{q + \frac{1}{2}} + O(1/q^2).
\end{align}
Therefore,
\begin{align}
    |x_{j + 1} - x_j| 
    &= \cos \left(\frac{\pi}{q + \frac{1}{2}} (j + 1 - \frac{1}{4})  \right) - \cos \left(\frac{\pi}{q + \frac{1}{2}} (j - \frac{1}{4})\right) + O(1/q^2) \\
    &= \frac{\pi}{q + \frac{1}{2}} \sin \left(\frac{\pi}{q + \frac{1}{2}}(j - \frac{1}{4})\right) + O(1/q^2) \\
    &= \frac{\pi}{q}\sqrt{1-x_j^2} + O(1/q^2).
\end{align}
\end{proof}

\weightbound*
\begin{proof}
    In Gauss-Legendre quadrature, weight $w_j$ corresponding to $j$-th roots $x_j$ of Legendre polynomial $P_q'(x)$ of degree $q$ is
    \begin{align}
    \label{eq9}
        w_j = \frac{2}{(1-x_j^2)P_n'(x_j)^2}.
    \end{align}
    For Legendre polynomial $P_q'(x)$ of degree $q$, by \cite[Theorem 8.21.2]{szego75}, asymptotically, when $\cos\theta$ is close to the roots, one has
    \begin{align}
        P_q(\cos \theta) = \sqrt{\frac{2}{\pi q \sin \theta}} \cos \left(\left(q + \frac{1}{2}\right)\theta - \frac{\pi}{4}\right) + O(q^{-3/2}).
    \end{align}
    Here $\cos \theta = x$, and we define $\cos\theta_j = x_j$ as the $j$-th root of our Legendre polynomial. Then
    \begin{align}
        \frac{\mathrm{d}}{\mathrm{d}\theta}P_q(\cos \theta)\mid_{\theta = \theta_j} 
        &= \frac{\mathrm{d}P_q(\cos \theta)}{\mathrm{d}(\cos\theta)} \frac{\mathrm{d}(\cos\theta)}{\mathrm{d}\theta}\mid_{\theta = \theta_j} \\
        &= \sqrt{\frac{2}{\pi q \sin \theta_j}} \cdot \left(q + \frac{1}{2}\right)\sin \left(\left(q + \frac{1}{2}\right)\theta_j - \frac{\pi}{4}\right) + O(q^{-3/2})\\
        &= \sqrt{\frac{2}{\pi q \sin \theta_j}} \cdot \left(q + \frac{1}{2}\right) + O(q^{-3/2}).
    \end{align}
    Here we have used 
    \begin{align}
        \cos \left(\left(q + \frac{1}{2}\right)\theta_j - \frac{\pi}{4}\right) = 0.
    \end{align}
    Therefore, one has
    \begin{align}
        \frac{\mathrm{d}P_q(\cos \theta)}{\mathrm{d}(\cos\theta)} \mid_{\theta = \theta_j} = \sqrt{\frac{2}{\pi q \sin^3 \theta_j}} \cdot \left(q + \frac{1}{2}\right) + O(1/q^3).
    \end{align}
    Then 
    \begin{align}
        P_q'(x_j)^2 = \frac{2}{\pi q \sin^3 \theta_j}\cdot \left(q + \frac{1}{2}\right)^2 + O(1/q^3),
    \end{align}
    here $x_j = \cos \theta_j$. By substituting back to \cref{eq9}, one has
    \begin{align}
        w_j 
        &= \frac{\pi q \sin \theta_j}{\left(q + \frac{1}{2}\right)^2} + O(1/q^3) \\
        &= \frac{\pi}{q} \sqrt{1-x_j^2} + O(1/q^2).
    \end{align}
    \end{proof}

\partialsum*
\begin{proof}
    By \cref{lemma:spacing,lemma:weight}, one has
\begin{align}
    \sum_{j = a}^b f(\hat x_j) w_j = \sum_{j = a}^b f(\hat x_j) (\hat x_{j + 1} - \hat x_j) + O((b-a)/q^2) = \sum_{j = a}^b f(\hat x_j) \delta + O((b-a)/q^2),
\end{align}
Where $\delta \coloneqq x_{j + 1} - x_j = O(1/q)$. By the Riemann sum error, one has
\begin{align}
    \left|\sum_{j=a}^b f(x_j) \delta - \int_{\hat x_a}^{\hat x_b} f(x) \mathrm{d}x \right| \leq  \frac{1}{2} \max_{x\in[-1,1]} |f'(x)| \sum_{i = a}^b \delta^2 = O(l(b-a)/q^2),
\end{align}
\end{proof}

\section{Proofs of technical lemmas}
\label{sec:proof-technical}

\duhamel*
\begin{proof}
    By Duhamel's principle, for a differential equation written in the form of,
    \begin{align}
    \label{eq:duhamel_DE}
        u'(t) = Lu(t) + f(t,u(t)), \quad u(0) = u_0,
    \end{align}
    where $L$ is a linear operator, one has solution of $u(t)$ with the following formula
    \begin{align}
    \label{eq:duhamel_SL}
        u(t) = e^{tL}u_0 + \int_0^t e^{(t-s)L} f(s,u(s))\mathrm{d}s.
    \end{align}
    If we let 
    \begin{align}
        u(t) = e^{-i(D + B)t}, \quad L= -iD, \quad f(t,u(t)) = -iB u(t), \quad u(0) = I,
    \end{align}
    we can write the differential equation \cref{eq:duhamel_DE} as:
    \begin{align}
        u'(t) = -iDu(t) -iBu(t), \quad u(0) = I.
    \end{align}
    Then, by Duhamel's principle, the solution \cref{eq:duhamel_SL} can be written as:
    \begin{align}
        u(t) &= e^{tL}u_0 + \int_0^t e^{(t-s)L} f(s,u(s))\mathrm{d}s\\
        &= e^{-iDt} + \int_0^t e^{-iD(t-s)} (-iB) u(t).
    \end{align}
    By substituting $u(t)$ with $e^{-i(D + B)t}$, one has
    \begin{align}
        e^{-i(D + B)t} = & e^{-iDt} + (-i) \int_0^t e^{-iD(t-s)} B e^{-i(D + B) s} \ \mathrm{d}s.
    \end{align}
\end{proof}

\fkbound*
\begin{proof}
If we further define 
\begin{align}
    J(s) \coloneqq e^{iDs}B e^{-iDs},
\end{align}
the above error may be rewritten as
\begin{align}
    \left\|\int_{0 \leq s_1 \leq \cdots \leq s_k \leq t} J(s_k) \cdots J(s_1) \, \mathrm{d}s_1 \cdots \mathrm{d}s_k - \sum_{j_1 = 1}^q \cdots \sum_{j_k = 1}^q J (\hat{x}_{(j_k)}) \cdots J(\hat{x}_{(j_k, \dots, j_1)}) \hat{w}_{(j_k)} \cdots \hat{w}_{(j_k, \dots, j_1)} \right\| 
\end{align}
We first discretize the last integral variable
\begin{align}
   &\|\int_{0\leq s_1 \leq \cdots \leq  s_k\leq t}J(s_k)\cdots J(s_1)\mathrm{d}s_1\cdots\mathrm{d}s_k\\
   & -\sum^q_{s_k = 1} \int_{0\leq s_1\leq \cdots\leq s_{k-1}\leq \hat{x}_{(s_k)}} J(\hat{x}_{(s_k)})J(s_{k-1})\cdots J(s_1)\hat{w}_{(s_k)}\dd s_1\cdots \dd s_{k-1} \| \\
   &\leq \|\int^t_0\dd s_k J(s_k) - \sum^q_{s_k = 1} J(\hat{x}_{(s_k)})\hat{w}_{(s_k)}\|O(\|B\|^{k-1} \frac{t^{k-1}}{(k-1)!})\\
   &\leq O(\frac{\|J^{(2q)}(t_{\max})\|t^{2q+1}q}{(2q)!2^{4q-1}}\|B\|^{k-1}\frac{t^{k-1}}{(k-1)!})\\
   &= O(\frac{\|D\|^{2q}\|B\|^{k} t^{2q}q}{(2q)!}\frac{t^{k-1}}{(k-1)!})
\end{align}
where we have used the relation $\max_{s\in [0,t]}\|J(s)\| = \|B\|$, defined as the largest spectral norm over time. And
\begin{align}
    \|J^{(2q)}(t_{\max})\| = (2\|D\|)^{2q}\|B\|
\end{align}
We further treat each $\hat{x}_{(j_k)}$ as an upper bound
\begin{align}
   &\|\sum^q_{j_k = 1} \int_{0\leq s_1\leq \cdots\leq s_{k-1}\leq \hat{x}_{(s_k)}} J(\hat{x}_{(s_k)})J(s_{k-1})\cdots J(s_1)\hat{w}_{(s_k)}\dd s_1\cdots \dd s_{k-1} \\
   &- \sum^q_{s_{k-1} = 1}\sum^q_{j_k = 1} \int_{0\leq s_1\leq \cdots \leq s_{k-2}\leq \hat{x}_{(s_k,j_{k-1})}}J(\hat{x}_{(j_k)})J(\hat{x}_{(j_k,j_{k-1})})J(s_{k-2})\cdots
    J(s_1)\hat{w}_{(s_k)}\hat{w}_{(s_k,j_{k-1})}\dd s_1 \cdots \dd s_{k-2}\|\\
   &\leq O\left(\left(\frac{t^{k-2}}{(k-2)!} \sum^q_{j_k = 1}\hat{w}_{(j_k)}\right) \frac{\|D\|^{2q}\|B\|^{k} t^{2q}q}{(2q)!}\right)
   = O\left(\frac{t^{k-1}}{(k-2)!} \frac{\|D\|^{2q}\|B\|^{k} t^{2q}q}{(2q)!}\right).
\end{align}
Doing this repeatedly, we will get
\begin{align}
   &\|\int_{0\leq s_1 \leq \cdots \leq  s_k\leq t}J(s_k)\cdots J(s_1)\mathrm{d}s_1\cdots\mathrm{d}s_k - \sum^q_{j_1 = 1}\cdots \sum^q_{j_k = 1}J(\hat{x}_{(j_k)})\cdots J(\hat{x}_{(j_k,\dots,j_1)})\hat{w}_{(j_k)}\cdots \hat{w}_{(j_k,\dots , j_1)} \| \\
   &\leq O\left(\left(\sum^{k-2}_{l=-1}\frac{t^{k-1}}{(k-l-2)!(l+1)!}\right)\frac{\|D\|^{2q}\|B\|^{k} t^{2q}q}{(2q)!}\right)
   = O\left(\frac{(2t)^{k-1}}{(k-1)!} \frac{\|D\|^{2q}\|B\|^{k} t^{2q}q}{(2q)!}\right).
\end{align}
\end{proof}

\stateprep*
\begin{proof}
   By \cref{lemma:partial-sum}, the partial sums of the squares of the amplitudes can be efficiently estimated. Then, we use the Grover-Rudolph method~\cite{GR02} to efficiently prepare the state. As each partial-sum estimation has errors, we prove the robustness of the Grover-Rudolph method in our case. First of all, notice that the full sum $\sum_{j=1}^q x_j^{\ell} w_j = t^{\ell+1}/(\ell+1)$ is some constant when $t = O(1)$. In the recursive structure of this state preparation approach, the numerator of the current layer will be canceled by the denominator of the previous layer. Let $S(a, b)$ denote the partial sum from $j=1$ to $j=b$ estimated by \cref{lemma:partial-sum}. In the end, the amplitude of the $k$-th entry is
   \begin{align}
       \label{eq:k-amplitude}
       \frac{1}{\sqrt{S(1, q/2) + S(q/2+1, q)}} \frac{\sqrt{S(k, k+1)}}{\sqrt{x_k^{\ell}w_k + x_{k+1}^{\ell}w_{k+1}}} \sqrt{x_k^{\ell}w_k}.
   \end{align}
   Here in the last layer, we used the exact values in the denominator to ensure the unitarity of the rotation gates. In \cref{eq:k-amplitude}, $\sqrt{S(1, q/2) + S(q/2+1, q)}$ is $O(\ell/q)$-away from the exact full sum, which is some constant. We also have
   \begin{align}
       S(k, k+1) = x_k^{\ell}w_k + x_{k+1}^{\ell}w_{k+1} + O(\ell/q^2)
   \end{align}
   because of \cref{lemma:partial-sum}. We have
   \begin{align}
       \left|\frac{1}{\sqrt{S(1, q/2) + S(q/2+1, q)}} \left(\frac{\sqrt{S(k, k+1)}}{\sqrt{x_k^{\ell}w_k + x_{k+1}^{\ell}w_{k+1}}} -1\right) \sqrt{x_k^{\ell}w_k} \right| \leq O(\sqrt{\ell}/q).
   \end{align}
   The above gives an upper bound on the $\infty$-norm, which further implies that
   \begin{align}
               \left\||\psi\rangle - \frac{1}{\sqrt{\sum_{j = 1}^q  x_{j}^l w_j}}\sum_{j = 1}^q  \sqrt{x_{j}^l w_j} |j\rangle \right\| \leq O(\sqrt{\ell/q}).
   \end{align}
   
   Now given that $\ell = O(\log (1/\epsilon))$, choosing $q = \Omega\left(\frac{1}{\epsilon^2}\log \frac{1}{\epsilon}\right)$ yields the desired error bound $O(\epsilon)$. 
\end{proof}

\end{document}